\documentclass[english]{article}
\usepackage{amsmath}
\usepackage{amsthm}
\usepackage{mathptmx}
\usepackage{newtxmath}
\usepackage[T1]{fontenc}
\usepackage[latin9]{inputenc}
\usepackage{xcolor}
\usepackage{latexsym}
\usepackage{prettyref}
\usepackage{float}
\usepackage{enumitem}
\usepackage{graphicx}
\usepackage{setspace}
\usepackage[authoryear]{natbib}
\PassOptionsToPackage{normalem}{ulem}
\usepackage{ulem}
\makeatletter

\providecolor{lyxadded}{rgb}{0,0,1}
\providecolor{lyxdeleted}{rgb}{1,0,0}

\DeclareRobustCommand{\lyxsout}[1]{\ifx\\#1\else\sout{#1}\fi}

\numberwithin{equation}{section}
  \theoremstyle{plain}
  \newtheorem*{prop*}{\protect\propositionname}
\theoremstyle{plain}
\newtheorem{thm}{\protect\theoremname}[section]
  \theoremstyle{plain}
  \newtheorem{prop}[thm]{\protect\propositionname}
  \theoremstyle{plain}
  \newtheorem{lem}[thm]{\protect\lemmaname}

\@ifundefined{date}{}{\date{}}
\RequirePackage[colorlinks,citecolor=blue]{hyperref}

\usepackage{dsfont}
\usepackage[all]{hypcap}
\usepackage{graphicx}
\usepackage{verbatim}

\newrefformat{def}{\hyperref[#1]{Definition \ref{#1}}}%
\newrefformat{prop}{\hyperref[#1]{Proposition \ref{#1}}}%
\newrefformat{lem}{\hyperref[#1]{Lemma \ref{#1}}}%
\newrefformat{thm}{\hyperref[#1]{Theorem \ref{#1}}}%
\newrefformat{cor}{\hyperref[#1]{Corollary \ref{#1}}}%
\newrefformat{exa}{\hyperref[#1]{Example \ref{#1}}}%
\newrefformat{eq}{\hyperref[#1]{(\ref{#1})}}%
\newrefformat{tab}{\hyperref[#1]{Table \ref{#1}}}%
\newrefformat{fig}{\hyperref[#1]{Figure \ref{#1}}}%
\newrefformat{sec}{\hyperref[#1]{Section \ref{#1}}}%
\newrefformat{assu}{\hyperref[#1]{Assumption \ref{#1}}}%
\newrefformat{alg}{\hyperref[#1]{Algorithm \ref{#1}}}%
\newrefformat{rem}{\hyperref[#1]{Remark \ref{#1}}}%

\AtBeginDocument{
  
}

\makeatother

\usepackage{babel}
  \providecommand{\lemmaname}{Lemma}
  \providecommand{\propositionname}{Proposition}
\providecommand{\theoremname}{Theorem}

\begin{document}

\title{Non-parametric Archimedean generator estimation with implications
for multiple testing}

\author{André Neumann\thanks{University of Bremen, Institute for Statistics, Bibliothekstraße 1,
D-28359 Bremen, Germany, Email: neumann@uni-bremen.de}~ and Thorsten Dickhaus\thanks{University of Bremen, Institute for Statistics, Bibliothekstraße 1,
D-28359 Bremen, Germany, Email: dickhaus@uni-bremen.de}}
\maketitle
\begin{abstract}
In multiple testing, the family-wise error rate can be bounded under
some conditions by the copula of the test statistics. Assuming that
this copula is Archimedean, we consider two non-parametric Archimedean
generator estimators. More specifically, we use the non-parametric
estimator from \citet{Genest_Neslehova_Ziegel_2011} and a slight
modification thereof. In simulations, we compare the resulting multiple
tests with the Bonferroni test and the multiple test derived from
the true generator as baselines.
\end{abstract}

\section{Introduction}

\textquotedbl{}De copulis non est disputandum!\textquotedbl{} \citep{okhrin2010}.
Indeed, the utilization of copulas for the construction and the mathematical
analysis of multiple tests has meanwhile become an established technique
in simultaneous statistical inference; see \citet{stp-copulae}, \citet{Bodnar-Dickhaus-EJS},
\citet{Schmidt2014a}, \citet{Schmidt2014b}, \citet{stange2015uncertainty},
\citet{Cerqueti-Lupi2018}, \citet{Neumann_Bodnar_Pfeifer_Dickhaus_2019},
and Sections 2.2.4 and 4.4 of \citet{Dickhaus-Buch2014}. In particular,
\citet{stp-copulae} have explicitly shown that the copula approach
leads to the most general construction method for single-step multiple
tests under known univariate marginal null distributions of the test
statistics. This is due to Sklar's Theorem (see \citet{Sklar1959}),
which implies that any dependency structure among test statistics
can be expressed by a copula. Thus, by standardizing the marginal
null distributions (i.e., by transforming test statistics into $p$-values),
the copula carries the complete distributional information which is
necessary for the calibration of a (multivariate) multiplicity-adjusted
rejection threshold or local significance level, respectively.

One important class of copula functions is constituted by so-called
Archimedean copulas (cf., e.g., Section 6 of \citet{embrechts-copulae-survey},
Section 5.4 of \citet{Q-Risk-Management}, Section 3 of \citet{Cerqueti-Lupi2018},
and \citet{McNeil2009} for the relevance of Archimedean copulas in
various applications). Such copulas are defined by a generator function
$\phi$. In the case that $\phi$ is a completely monotone function,
the corresponding Archimedean copula possesses the property of being
multivariate totally positive of order $2$ (MTP$_{2}$), see \citet{Scarsini2005}.
This positive dependency property is very helpful for proving type
I error control of a variety of multiple tests, including the extremely
popular linear step-up test by \citet{benjamini1995controlling} for
control of the false discovery rate. \citet{Bodnar-Dickhaus-EJS}
have derived an adjustment factor for optimizing the power of the
linear step-up test in the case that the copula of the $p$-values
is of Archimedean type.

In practice, the copula of test statistics or $p$-values, respectively,
is often an unknown nuisance parameter (of potentially infinite dimension).
Therefore, it is near at hand to pre-estimate this copula and to utilize
the estimate in the construction of the rejection rule of the multiple
test. We call the resulting multiple test an empirically calibrated
multiple test. In the context of Archimedean copulas, \citet{stange2015uncertainty}
and \citet{Bodnar-Dickhaus-EJS} have investigated method-of-moments
and maximum likelihood estimators, respectively. Their examples, however,
are restricted to simple parametric copula families. \citet{Neumann_Bodnar_Pfeifer_Dickhaus_2019}
studied non-parametric Bernstein copula estimators with respect to
an empirical calibration of multiple tests. Their methodology, however,
is not specifically tailored towards Archimedean copula models.

In the present work, we therefore consider the problem of estimating
the generator function $\phi$ of an Archimedean copula non-parametrically,
and we analyze the impact of the estimation and its uncertainty on
the performance of empirically calibrated multiple tests. The material
is organized as follows. In \prettyref{sec: Section 2}, estimation
methods based on Kendall's processes, as well as a novel modification
thereof, are introduced. \prettyref{sec: Section 3} deals with multiple
testing methodology based on such estimators. Simulation results are
presented in \prettyref{sec: Simulations}, and we provide some conclusions
in \prettyref{sec: Conclusions}.

\section{Estimation of the Archimedean generator function using Kendall's
process\label{sec: Section 2}}

The application of Kendall's distribution functions for statistical
inference has been treated in quite some detail in \citet{Genest_Neslehova_Ziegel_2011}.
In this section, we shortly summarize their findings  and consider
additionally a slightly modified generator estimator. \citet{Genest_Neslehova_Ziegel_2011}
propose an algorithm for estimating an Archimedean copula based on
a slightly modified version of the estimated Kendall's distribution
function as defined by \citet{MR1405589}. In \citet{MR1405589} the
asymptotic properties of Kendall's process are analyzed. The estimator
of the Kendall's distribution function is modified such that these
properties still hold and the estimator itself is a Kendall's distribution
function. This avoids numerical issues and additionally, the continuous
mapping theorem (CMT) can be applied. The resulting copula estimator
is (weakly) consistent. We will see in the next section that this
translates under some conditions directly to the realized family-wise
error rate (FWER) in the case that the modified estimator is utilized
in the empirical calibration of a multiple single-step test.

Further, \citet{Genest_Neslehova_Ziegel_2011} discussed the identifiability
of statistical models based on Kendall's distribution function. In
particular, the two- and three-dimensional cases are proven. In the
two-dimensional case this property can be verified directly, i.e.,
the Archimedean generator can be directly calculated from Kendall's
distribution function. In the three-dimensional case they show that
the model is identified by reducing the dimensionality in a sense.
More precisely, they show that from Kendall's distribution function
$K_{\phi_{1},m=3}=K_{\phi_{2},m=3}$ follows $K_{\phi_{1},m=2}=K_{\phi_{2},m=2}$,
where $\phi_{1}$, $\phi_{2}$ are Archimedean generator functions
and $m$ is the dimension. Inspired by this argumentation, we consider
additionally an estimator based on averaging over two-dimensional
marginal data samples in our simulations, i.e.,
\[
\hat{\phi}_{n,m}={m \choose 2}^{-1}\sum_{j_{1}=1}^{m}\sum_{j_{2}=j_{1}+1}^{m}\hat{\phi}_{n,2}^{\text{GNZ}}\left(X_{i,j_{1}},X_{i,j_{2}}:1\leq i\leq n\right),
\]
where $\hat{\phi}_{n,m}^{\text{GNZ}}$ is the generator estimator
from \citet{Genest_Neslehova_Ziegel_2011} and $n$ is the sample
size.

This estimator can be theoretically motivated as follows. Let $C_{\phi}$
defined by $C_{\phi}\left(\boldsymbol{u}\right)=\psi\left(\sum_{j=1}^{m}\phi\left(u_{j}\right)\right)$,
$\boldsymbol{u}\in\left[0,1\right]^{m}$, be an $m$-dimensional Archimedean
copula, where the generator $\phi:\left[0,1\right]\rightarrow\mathbb{R}_{+}$
is a continuous, convex and strictly decreasing function with $\phi\left(1\right)=0$
and $\psi:v\mapsto\inf\left\{ u\in\left[0,1\right]|\phi\left(u\right)\leq v\right\} $
is the generalized inverse of $\phi$. For convenience, we will call
both - $\phi$ and $\psi$ - the Archimedean generator and $C_{\psi}=C_{\phi}$,
$K_{\psi}=K_{\phi}$ . Let $\phi_{n}$ be a sequence of generators
corresponding to given Kendall's distribution functions $K_{n,m=2}$
which converge weakly\footnote{A distribution function converges weakly ($\stackrel{w}{\rightarrow}$)
if and only if it converges pointwise for all continuity points of
the limit function.} to $K$. \citet{Genest_Neslehova_Ziegel_2011} used their Proposition
6 in order to translate the convergence of the generators to the convergence
of the Archimedean copulas. Let us restate this proposition:
\begin{prop*}[Proposition 6 in \citet{Genest_Neslehova_Ziegel_2011}]
Let $\psi$, $\psi_{1}$, $\psi_{2}$, ... be $m$-monotone\footnote{A generator $\psi$ is $m$-monotone if and only if it has $m-2$
derivatives such that $\left(-1\right)^{i}\psi^{\left(i\right)}\geq0$
for all $0\leq i\leq m-2$, and $\left(-1\right)^{m-2}\psi^{\left(m-2\right)}$
is decreasing and convex.} Archimedean generators. Then $C_{\psi_{n},m}\stackrel{w}{\rightarrow}C_{\psi,m}$,
as $n\rightarrow\infty$, if and only if there exists a sequence $\left(a_{n}\right)$
of constants $a_{n}>0$ such that, for any $x\in[0,\infty)$, $\psi_{n}\left(a_{n}x\right)\rightarrow\psi\left(x\right)$,
as $n\rightarrow\infty$.
\end{prop*}
Assume that $\psi_{n}$ is only $2$-monotone. Then $C_{\psi_{n},m}$
is not an Archimedean copula anymore. However, the argumentation for
the converse only uses the uniform convergence of $\psi_{n}$ and
$\phi_{n}$ for any dimension $m\geq2$. This means that the following
proposition holds:
\begin{prop}
Let $\psi$, $\psi_{1}$, $\psi_{2}$, ... be $2$-monotone Archimedean
generators. If there exists a sequence $\left(a_{n}\right)$ of constants
$a_{n}>0$ such that, for any $x\in[0,\infty)$, $\psi_{n}\left(a_{n}x\right)\rightarrow\psi\left(x\right)$,
as $n\rightarrow\infty$, then $C_{\psi_{n},m}\stackrel{w}{\rightarrow}C_{\psi,m}$,
as $n\rightarrow\infty$.
\end{prop}

Applying Algorithm 1 from \citet{Genest_Neslehova_Ziegel_2011} to
two-dimensional marginal data results in a $2$-monotone estimator
$\hat{\psi}_{n}$ (see Section 4.3 in \citet{Genest_Neslehova_Ziegel_2011}).
For a specific sequence $\left(a_{n}\right)_{n\in\mathbb{N}}$ this
estimator converges uniformly to the true generator $\psi$ on compact
sets in probability (see Section 7 in \citet{Genest_Neslehova_Ziegel_2011}).
Thus, $C_{\hat{\psi}_{n},m}\stackrel{w}{\rightarrow}C_{\psi,m}$ holds
in probability (and therefore, uniformly in probability since $C_{\psi,m}$
is continuous and $\left[0,1\right]$ is a compact interval).

On the other hand, for a $2$-monotone generator $\psi_{n}$ the Archimedean
copula restricted to the diagonal line $u\mapsto C_{\psi_{n},m}\left(u,\dots,u\right)=\psi_{n}\left(m\cdot\phi_{n}\left(u\right)\right)$
is still a one-dimensional distribution function. We will use this
function in the next section for calibrating the multiple test under
equal local significance levels. This means that we will receive valid
(i.e., in $\left[0,1\right]$) local significance levels.. Therefore,
the negative impact of using two-dimensional marginal data on multiple
testing is potentially small. We will explore this further in the
simulations of \prettyref{sec: Simulations}.    

\section{Multiple testing\label{sec: Section 3}}

In this section, we apply the Archimedean copula estimator of \citet{Genest_Neslehova_Ziegel_2011}
in a multiple testing context. The consistency of the realized FWER
is a direct consequence of the connection between the FWER and the
copula of the test statistics (see Theorem 2 in \citet{stp-copulae})
and the consistency of the estimator (see Theorem 2 in \citet{Genest_Neslehova_Ziegel_2011}).

\subsubsection*{Model}

We perform a multiple test $\boldsymbol{\varphi}=\left(\varphi_{1},\ldots,\varphi_{1}\right)^{\top}$
for $m\geq2$ null hypotheses $H_{j}\subset\boldsymbol{\Theta}$ versus
$K_{j}\subseteq\boldsymbol{\Theta}\backslash H_{j}$, $1\leq j\leq m$.
The parameter space $\boldsymbol{\Theta}=\Theta_{1}\times\ldots\times\Theta_{m}$
is a subset of $\mathbb{R}^{m}$. The marginal test $\varphi_{j}$
rejects the null hypothesis $H_{j}$ if and only if the corresponding
$p$-value $p_{j}$ is smaller than a local significance level $\alpha_{\text{loc}}^{\left(j\right)}\in\left(0,1\right)$.
Our goal is to find $\alpha_{\text{loc}}^{\left(1\right)},\ldots,\alpha_{\text{loc}}^{\left(m\right)}$
such that the FWER is controlled at the global significance level
$\alpha\in\left(0,1\right)$. A FWE occurs whenever we reject at least
one true null hypothesis.

\subsubsection*{Assumptions}

Let $\boldsymbol{\vartheta}=\left(\vartheta_{1},\ldots,\vartheta_{m}\right)^{\top}\in\boldsymbol{\Theta}$
be the parameter vector of interest and $1\leq j\leq m$. We assume
that the following statements hold true. 
\begin{enumerate}
\item The null hypothesis $H_{j}$ has the form $H_{j}:\left\{ \boldsymbol{\vartheta}\mid\vartheta_{j}\in\Theta_{j}^{*}\right\} $,
where $\emptyset\neq\Theta_{j}^{*}\subset\Theta_{j}\subseteq\mathbb{R}$.
\end{enumerate}
This means that we can only test marginal parameters. In particular,
the number of parameters is equal to the number of hypotheses $m$.
With a slight abuse of notation, $H_{j}$ denotes also $\Theta_{j}^{*}$.

We define the $p$-value $p_{j}$ by $1-\inf_{\boldsymbol{\vartheta}\in H_{j}}F_{T_{j}\mid\boldsymbol{\vartheta}}\left(T_{j}\right)$,
where $T_{j}$ is known.
\begin{enumerate}[resume]
\item The test statistic $T_{j}$ tends to larger values under the alternative
$K_{j}$.
\end{enumerate}
This means that the $p$-values tend to smaller values under alternatives.
\begin{enumerate}[resume]
\item The cdf $F_{T_{j}\mid\boldsymbol{\vartheta}}=F_{T_{j}\mid\vartheta_{j}}$
of$T_{j}$ depends only on $\vartheta_{j}$.
\end{enumerate}
We test the marginal parameter $\vartheta_{j}\in H_{j}$ by using
only the $p$-value $p_{j}$. It makes sense that $p_{j}$ does not
depend on the other parameters.
\begin{enumerate}[resume]
\item There exists a $\vartheta_{j}^{*}\in H_{j}$ such that $F_{T_{j}\mid\vartheta_{j}^{*}}\equiv\inf_{\vartheta_{j}\in H_{j}}F_{T_{j}\mid\vartheta_{j}}$.
\end{enumerate}
The important part here is that this equality holds for every point.
This means that the parameters $\vartheta_{1}^{*},\ldots,\vartheta_{m}^{*}$
are independent of the observed test statistics.
\begin{enumerate}[resume]
\item The cdf $F_{T_{j}\mid\vartheta_{j}^{*}}$ is continuous.
\end{enumerate}
Under this assumption, we have the following properties: The corresponding
$p$-value $p_{j}=1-F_{T_{j}\mid\vartheta_{j}^{*}}\left(T_{j}\right)$
is valid, i.e., $\text{Prob}\left(\left(p_{j}\mid H_{j}\right)\leq u\right)\leq u$,
and uniformly distributed under $\vartheta_{j}^{*}$. We write $Z\mid H_{j}$
or $Z\mid K_{j}$ for a random variable $Z=Z\mid\vartheta_{j}$ under
true null hypothesis $\vartheta_{j}\in H_{j}$ or under true alternative
$\vartheta_{j}\in K_{j}$, respectively. Further, we have almost surely
that $p_{j}<\alpha_{j}$, $\alpha_{j}\in\left[0,1\right]$, if and
only if $T_{j}>c_{j}\left(\alpha_{j}\right)$, where $c_{j}\left(\alpha_{j}\right)=F_{T_{j}\mid\vartheta_{j}^{*}}^{-1}\left(1-\alpha_{j}\right)$.
In addition, the copula of the test statistics under $\boldsymbol{\vartheta}^{*}$
is unique by Sklar's theorem.
\begin{enumerate}[resume]
\item The parameter vector $\boldsymbol{\vartheta}^{*}$ is a least favorable
configuration (LFC).
\end{enumerate}
An LFC is a parameter vector which maximizes the FWER. In our case
the LFC $\boldsymbol{\vartheta}^{*}$ is located in the global null
hypothesis $H_{0}=\bigcap_{j=1}^{m}H_{j}$. This means that weak FWER
control (i.e., control on $H_{0}$) entails strong FWER control (i.e.,
control on $\boldsymbol{\Theta}$). In our setting the subset pivotality
condition (see \citet{westfallyoung}) is sufficient for this assumption
(see Lemma 3.1 in \citet{Dickhaus_Stange_2013} and Lemma 3.3 in \citet{Neumann_Bodnar_Pfeifer_Dickhaus_2019}).
\begin{enumerate}[resume]
\item The copula of the test statistics $C_{\boldsymbol{T}\mid\boldsymbol{\vartheta}^{*}}$
under the LFC \textbf{$\boldsymbol{\vartheta}^{*}$ }is Archimedean.
\end{enumerate}
Note that $C_{\boldsymbol{T}\mid\boldsymbol{\vartheta}^{*}}$ can
depend on the whole parameter vector $\boldsymbol{\vartheta}^{*}$.
\begin{enumerate}[resume]
\item A pseudo sample of the test statistics under $\boldsymbol{\vartheta}^{*}$
is at hand.
\end{enumerate}
This sample is needed to calibrate the multiple test and can be (approximately)
achieved by bootstrapping under $\boldsymbol{\vartheta}^{*}$ (see
\citet{MR515681}) for suitable test statistics.

\subsubsection*{Calibration and consistency}

Under some assumptions, \citet{stp-copulae} have shown that the
FWER is uniformly bounded on the parameter space $\boldsymbol{\Theta}$
by the copula of the random vector $\left(1-p_{1},\ldots,1-p_{m}\right)$
under the LFC $\boldsymbol{\vartheta}^{*}$. In our setup this is
the same copula as the copula of the test statistics. More precisely,
under our assumptions we have
\begin{align*}
\text{FWER}_{\boldsymbol{\vartheta}}\left(\varphi_{1},\ldots,\varphi_{m}\right) & \leq\text{FWER}_{\boldsymbol{\vartheta}^{*}}\left(\varphi_{1},\ldots,\varphi_{m}\right)\\
 & =\text{Prob}\left(\bigcup_{j=1}^{m}\left\{ \left(p_{j}\mid\vartheta_{j}^{*}\right)<\alpha_{\text{loc}}^{\left(j\right)}\right\} \right)\\
 & =1-\text{Prob}\left(\bigcap_{j=1}^{m}\left\{ 1-\left(p_{j}\mid\vartheta_{j}^{*}\right)\leq1-\alpha_{\text{loc}}^{\left(j\right)}\right\} \right)\\
 & =1-C_{\boldsymbol{T}\mid\boldsymbol{\vartheta}^{*}}\left(1-\alpha_{\text{loc}}^{\left(1\right)},\ldots,1-\alpha_{\text{loc}}^{\left(m\right)}\right).
\end{align*}
The bound depends on the copula of the test statistics under the LFC
$\boldsymbol{\vartheta}^{*}$. This has the advantage that we only
need to estimate the Archimedean copula under one parameter vector
in order to calibrate the multiple test. Setting $\alpha_{\text{loc}}=\alpha_{\text{loc}}^{\left(1\right)}=\ldots=\alpha_{\text{loc}}^{\left(m\right)}$,
the estimated local significance level $\hat{\alpha}_{\text{loc},n}$
is given by
\[
\hat{\alpha}_{\text{loc},n}=1-\hat{C}_{\boldsymbol{T},n}^{-1}\left(1-\alpha\right),
\]
where $\hat{C}_{\boldsymbol{T},n}^{-1}$ is the quantile of the estimated
Archimedean copula on the diagonal line $u\mapsto\hat{C}_{\boldsymbol{T},n}\left(u,\ldots,u\right)=\hat{\psi}_{n}\left(m\cdot\hat{\phi}_{n}\left(u\right)\right)$.
For convenience, we will write occasionally $C\left(u\right)$ instead
of $C\left(u,\ldots,u\right)$. Note that $\hat{C}_{\boldsymbol{T},n}^{-1}\left(v\right)=\hat{\psi}_{n}\left(m^{-1}\cdot\hat{\phi}_{n}\left(v\right)\right)$
for $v\in(0,1]$ and zero otherwise. This means that $\hat{\alpha}_{\text{loc},n}$
can be drawn in the graph of $\hat{\phi}_{n}$ (see \prettyref{fig: alpha_loc}).
\begin{figure}[H]
\begin{centering}
\includegraphics[width=0.5\columnwidth]{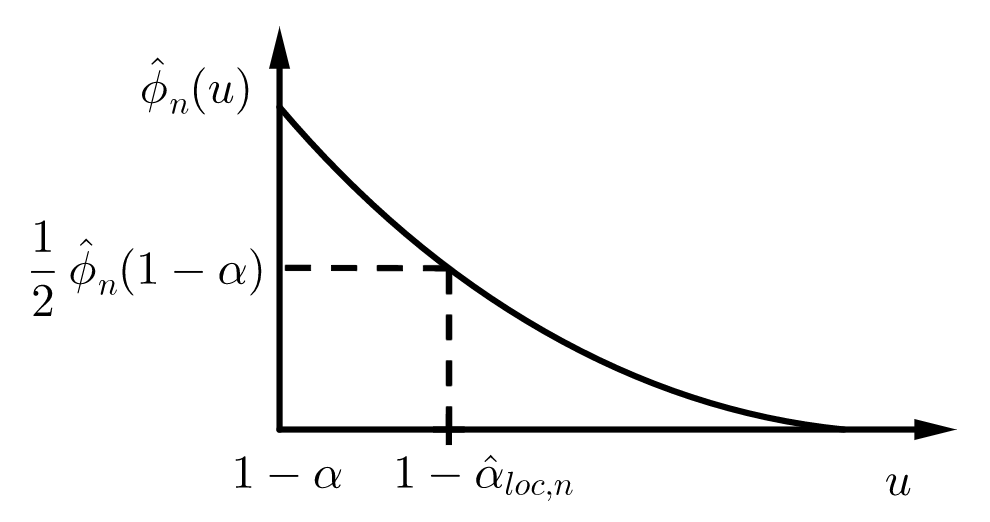}
\par\end{centering}
\caption{The estimated local significance level $\hat{\alpha}_{\text{loc},n}$
for $m=2$ drawn in the graph of the estimated generator $\hat{\phi}_{n}$
.\label{fig: alpha_loc}}
\end{figure}
\begin{lem}
\label{lem: hat_C_T is uniformly consistent}The quantile function
$\hat{C}_{\boldsymbol{T},n}^{-1}$ is uniformly consistent, i.e.,
it converges uniformly to $C_{\boldsymbol{T}\mid\boldsymbol{\vartheta}^{*}}^{-1}$
in probability as $n\rightarrow\infty$. In particular, $\hat{\alpha}_{\text{loc},n}$
is consistent.
\end{lem}

\begin{proof}
The argumentation is the same as for the consistency of $\hat{\psi}_{n}$
in \citet{Genest_Neslehova_Ziegel_2011}. Let $C_{n}$, $n\in\mathbb{N}$,
be non-decreasing functions with $C_{n}\stackrel{w}{\rightarrow}C$.
Then $C_{n}^{-1}\stackrel{w}{\rightarrow}C^{-1}$ (see Proposition
0.1 in \citet{Resnick_1987}). Again, $C^{-1}$ denotes the quantile
of $u\mapsto C\left(u,\ldots,u\right)$. Thus, the function $C\mapsto C^{-1}$
is continuous on the space of cdfs with a suitable norm for the weak
convergence. By the CMT, $\hat{C}_{\boldsymbol{T},n}^{-1}$ is pointwise
consistent on the compact set of continuity points $(0,1]$ and therefore,
uniformly on $\left[0,1\right]$.
\end{proof}
\begin{prop}
\label{prop: rFWER is consistent}The realized FWER defined by $\widehat{\text{rFWER}}_{n}=1-C_{\boldsymbol{T}\mid\boldsymbol{\vartheta}^{*}}\left(1-\hat{\alpha}_{\text{loc},n}\right)$
converges in probability to $\alpha$ as $n\rightarrow\infty$.
\end{prop}

\begin{proof}
This result follows directly from \prettyref{lem: hat_C_T is uniformly consistent}
and the CMT applied to $u\mapsto C_{\boldsymbol{T}\mid\boldsymbol{\vartheta}^{*}}\left(u,\ldots,u\right)$.
\end{proof}

\section{Simulations\label{sec: Simulations}}

\subsubsection*{Overview}

In this section we consider three simulations. The first simulation
compares pointwise the estimated copulas on the diagonal line. The
second simulation is more of theoretic interest. There, we sample
the test statistics under the LFC directly and continue as in the
first simulation. We compare the estimated generators in terms of
their calibrated local significance level and empirical power. In
the last simulation we consider a misspecified setting by sampling
normally distributed data. We use the sample means as test statistics
and bootstrap to derive a pseudo sample for the test statistics under
the LFC.

\subsubsection*{Simulation 1 - comparison of the generators}

In this simulation a sample of the test statistics is directly taken
from a Gumbel copula under various parameters derived from Kendall's
$\tau$. For consistency with Simulation 3, we denote the sample size
of the test statistics by $B$ instead of $n$. We consider the estimators
$\hat{\phi}_{B,m}^{\text{GNZ}}$ from \citet{Genest_Neslehova_Ziegel_2011}
and $\hat{\phi}_{B,m}^{\text{MC}}$, which is defined by 
\[
\hat{\phi}_{B,m}^{\text{MC}}=M^{-1}\sum_{\ell\in M}\hat{\phi}_{B,2}^{\text{GNZ}}\left(X_{i,j_{1}\left(\ell\right)},X_{i,j_{2}\left(\ell\right)}:1\leq i\leq B\right),
\]
where $1\leq j_{1}\neq j_{2}\leq m$, $m$ is the dimension and $M=100$
is the number of Monte Carlo repetitions. In low dimensions it is
possible to average over all two-dimensional subsamples as in \prettyref{sec: Section 2}.
We compare the estimated copulas regarding the pointwise absolute
distance to the true copula $C_{\phi}$, i.e., $\left|C_{\hat{\phi}_{B}}\left(u,\ldots,u\right)-C_{\phi}\left(u,\ldots,u\right)\right|$,
$u\in\left\{ 0,1/\left(B+1\right),\ldots,1\right\} $.

The default setting is $\tau=0.5$, $B=100$ and $m=6$. In \prettyref{fig: Simulation 1-1}
and \prettyref{fig: Simulation 1-2}, we plot the resulting Archimedean
copulas $C\left(u,\ldots,u\right)$, where $u$ lies on the grid $\left\{ 0,1/\left(B+1\right),\ldots,1\right\} $.
At each point the copula is averaged over $L$ repetitions. It can
be generally observed that the copula of both estimators are comparable
with some exceptions. Under independence, GNZ performs considerably
worse. In dimensions $m\in\left\{ 30,60\right\} $, the distance between
MC and the true copula is slightly smaller and more consistent compared
to GNZ. Further, it can be observed that the sample size should not
be too small ($B\geq100$). There is a slight but visible improvement
from $B=50$ to $B=100$. We will consider additional settings in
the next two simulations. 
\begin{figure}[H]
\begin{centering}
\includegraphics[width=0.4\columnwidth]{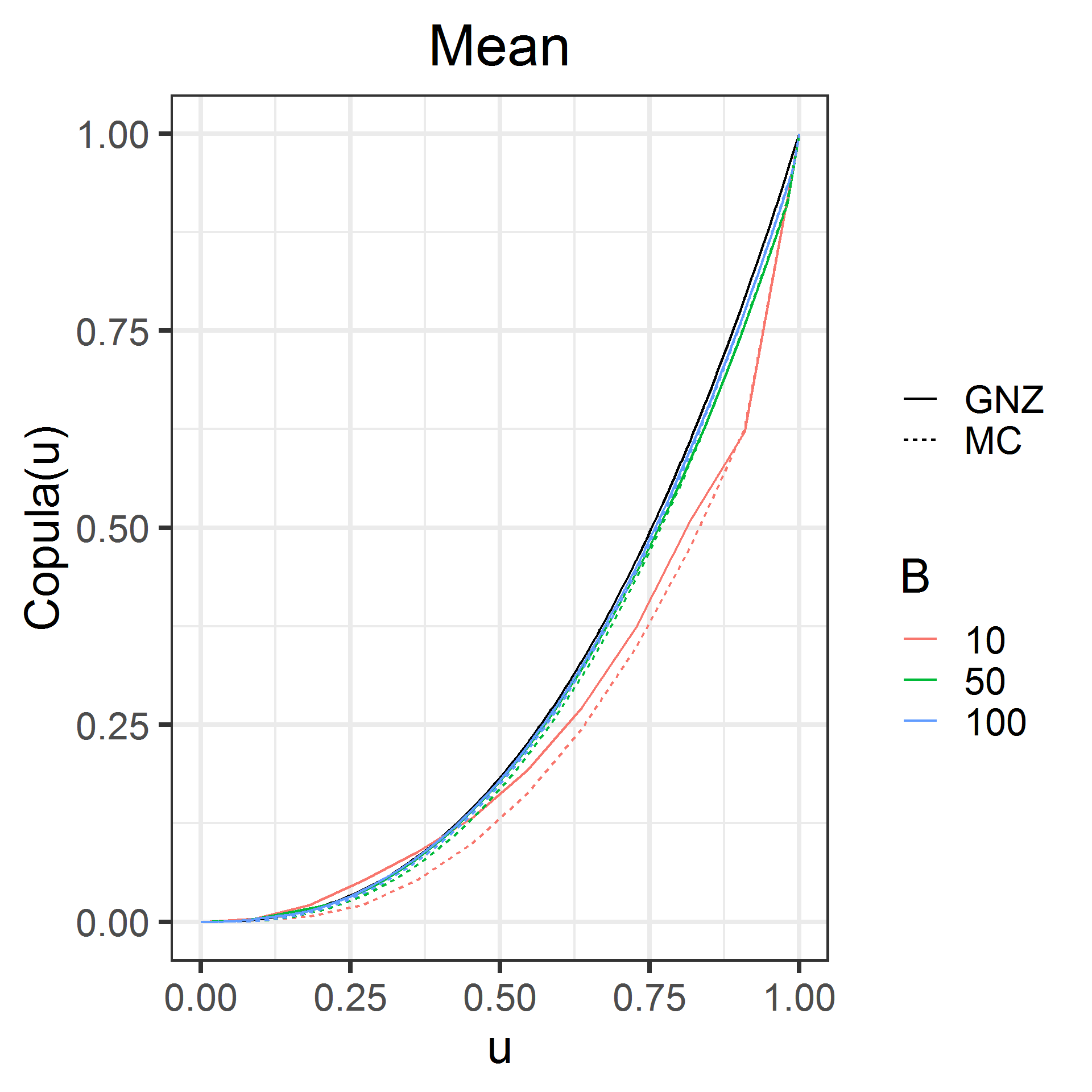}$\qquad$\includegraphics[width=0.4\columnwidth]{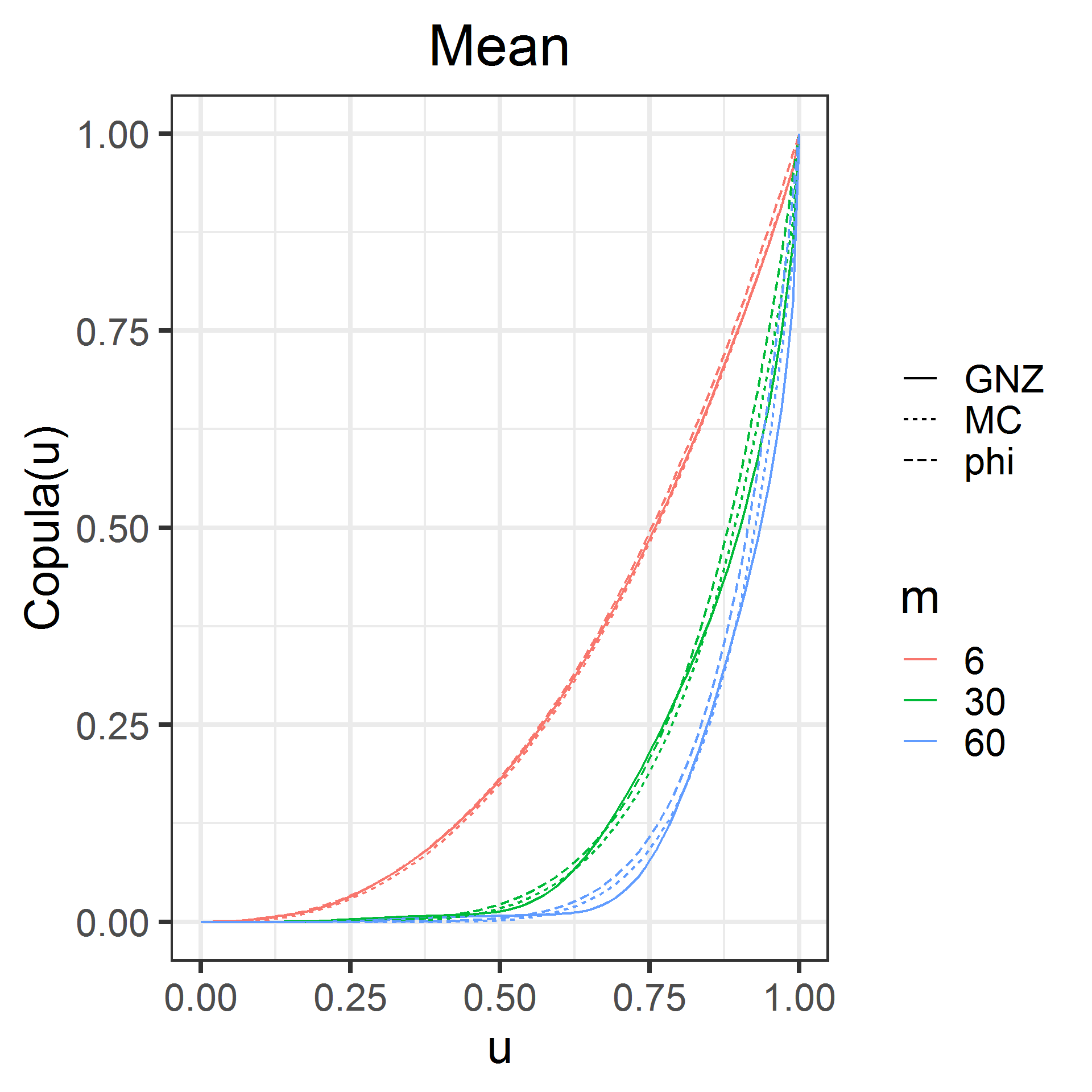}
\par\end{centering}
\caption{The pointwise mean of the copulas over $L=1{,}000$ repetitions with
sample size $B\in\left\{ 10,50,100\right\} $ (left) and dimension
$m\in\left\{ 6,30,60\right\} $ (right). The black line in the left
plot corresponds to the true copula $C_{\phi}$. All other parameters
are set to default.\label{fig: Simulation 1-1}}
\end{figure}
\begin{figure}[H]
\begin{centering}
\includegraphics[width=0.4\columnwidth]{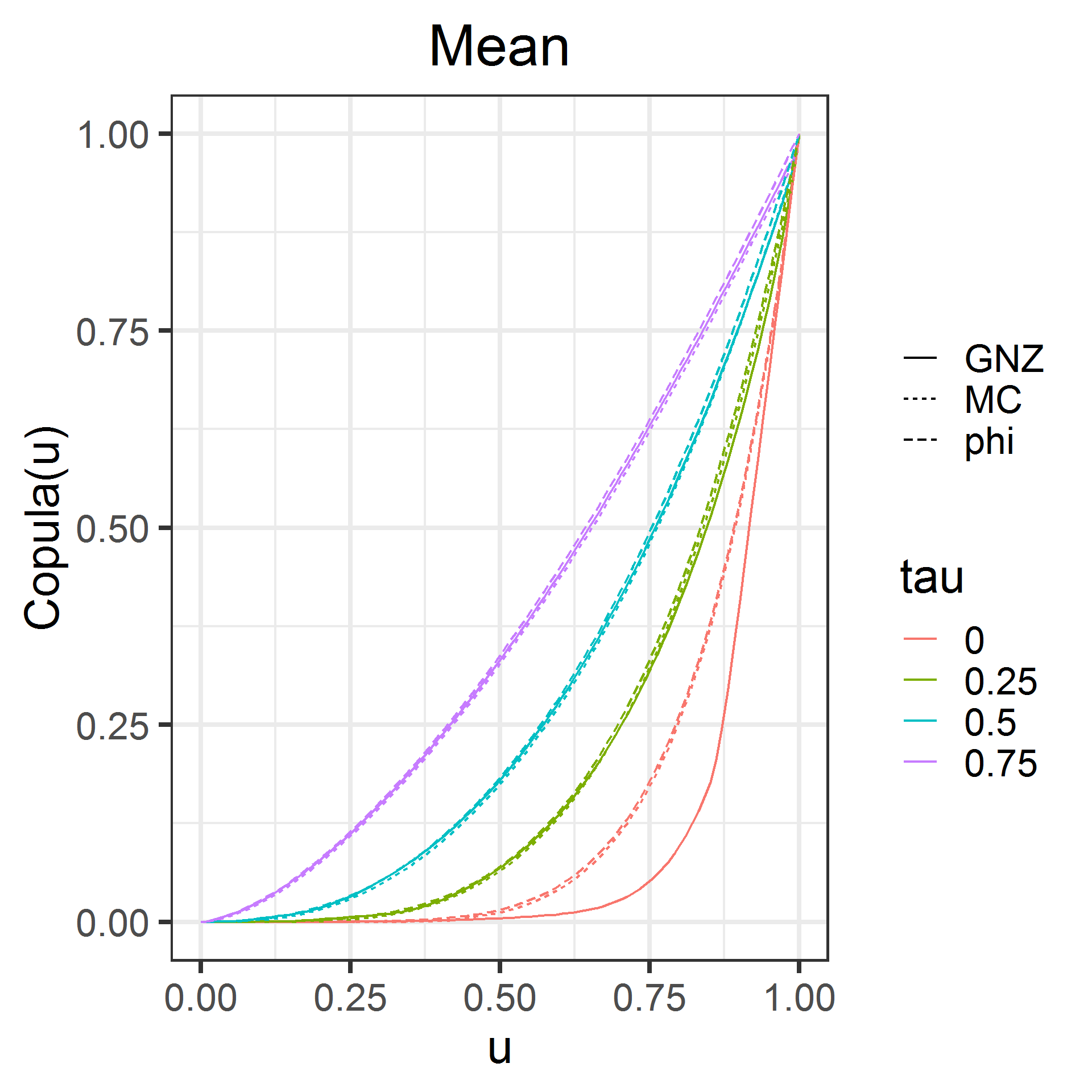}
\par\end{centering}
\caption{The pointwise mean of the copulas over $L=1{,}000$ repetitions with
Kendall's $\tau\in\left\{ 0,0.25,0.5,0.75\right\} $. All other parameters
are set to default.\label{fig: Simulation 1-2}}
\end{figure}

\subsubsection*{Simulation 2 - calibration in an artificial setting}

The setting in this simulation is very similar to Simulation 1. Again,
no raw data is simulated. Instead, a sample of the test statistics
$\boldsymbol{T}_{1}^{*},\ldots,\boldsymbol{T}_{B}^{*}$ from a Gumbel
copula $C_{\theta^{*}}$, $\theta^{*}\geq1$, with standard normal
marginals is given, where $\theta^{*}$ corresponds to the copula
parameter under the LFC. We use this sample to calibrate the multiple
test. This means that the estimation procedure from the first simulation
is applied to $\boldsymbol{T}_{1}^{*},\ldots,\boldsymbol{T}_{B}^{*}$.
Then, we derive the estimated local significance level $\hat{\alpha}_{\text{loc},B}$
as in \prettyref{sec: Section 3}.

We test the mean vector $\boldsymbol{\vartheta}=\boldsymbol{\mu}$
by the hypotheses $H_{j}:\left\{ \mu_{j}=0\right\} $ versus $K_{j}:\left\{ \mu_{j}\neq0\right\} $.
The observed test statistics $\boldsymbol{T}=\left(T_{1},\ldots,T_{m}\right)^{\top}$
are taken from a Gumbel copula $C_{\theta}$, $\theta\geq1$, with
normal marginals. More specifically, $T_{j}\mid H_{j}\sim\mathcal{N}\left(0,n^{-1}\right)$
and $T_{j}\mid K_{j}\sim\mathcal{N}\left(\mu,n^{-1}\right)$ , where
$\mu>0$ denotes the unknown global effect size (for simplicity) and
$n\in\mathbb{N}$ emulates a effect of the data sample size. The copula
parameter $\theta^{*}=\theta^{*}\left(\tau\right)$ is directly derived
from Kendall's $\tau$. We set $\theta=\theta^{*}+\left(1-\pi_{0}\right)\mu$,
where $\pi_{0}$ is the proportion of true null hypotheses. The two-sided
$p$-value $p_{j}$ is given by $2\left(1-F_{T_{j}\mid H_{j}}\left(\left|T_{j}\right|\right)\right)$.

In this setup, the copula of the test statistics depends indirectly
on $\boldsymbol{\mu}$. In particular, we have $\theta=\theta^{*}$
under $H_{0}$. Note that $0\equiv\boldsymbol{\mu}^{*}\in H_{0}$
is an LFC since the Gumbel copula on the diagonal line $u\mapsto C_{\theta}\left(u,\ldots,u\right)=u^{m^{1/\theta}}$
is non-increasing for $\theta\downarrow\theta^{*}$, where $m>1$
and $0\leq u\leq1$. Therefore, it is straightforward to check that
$\text{FWER}_{\boldsymbol{\mu}}\left(\boldsymbol{\varphi}\right)\leq1-C_{\theta}\left(1-\alpha_{\text{loc}}\right)\leq1-C_{\theta^{*}}\left(1-\alpha_{\text{loc}}\right)=\text{FWER}_{\boldsymbol{\mu}^{*}}\left(\boldsymbol{\varphi}\right)$. 

The default setting is $\tau=0.5$, $B=100,$ $n=100$, $m=6$ and
$\mu=0.2$. In \prettyref{fig: Simulation 2-1} to \prettyref{fig: Simulation 2-3},
the estimated local significance level and empirical power are compared,
respectively. The empirical FWER is omitted but always below $5\%$.
 With increasing sample size $B$, both $\hat{\alpha}_{\text{loc},B}^{\text{GNZ}}$
and $\hat{\alpha}_{\text{loc},B}^{\text{MC}}$ converge to a local
significance level close to the true significance level $\alpha_{\text{loc}}^{\phi_{\theta}}$.
With larger dimension $m$, $\hat{\alpha}_{\text{loc},B}^{\text{GNZ}}$
and $\hat{\alpha}_{\text{loc},B}^{\text{MC}}$ fall drastically and
are close to the Bonferroni correction.  Under stronger dependence
$\tau$, $\hat{\alpha}_{\text{loc},B}^{\text{GNZ}}$ and $\hat{\alpha}_{\text{loc},B}^{\text{MC}}$
deviate again from the Bonferroni correction and get closer to $\alpha_{\text{loc}}^{\phi_{\theta}}$.
In most settings, the GNZ estimator is slightly better in terms of
the mean local significance level but slightly worse in terms of standard
deviation (see the supplement for the standard deviation plots and
numerical results). Overall , this simulation suggests that both estimators
are comparable for calibrating multiple tests.
\begin{figure}[H]
\begin{centering}
\includegraphics[width=0.4\columnwidth]{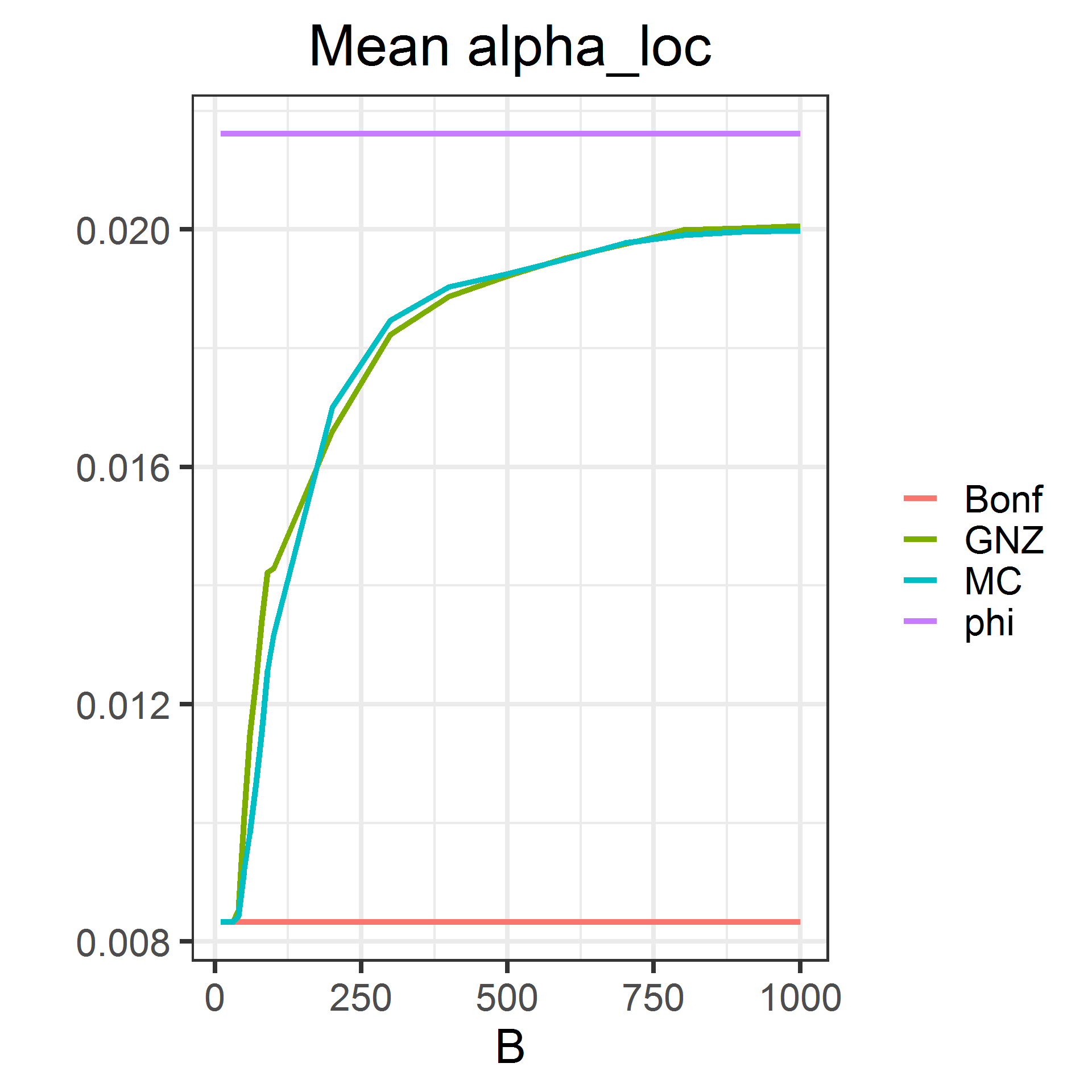}$\qquad$\includegraphics[width=0.4\columnwidth]{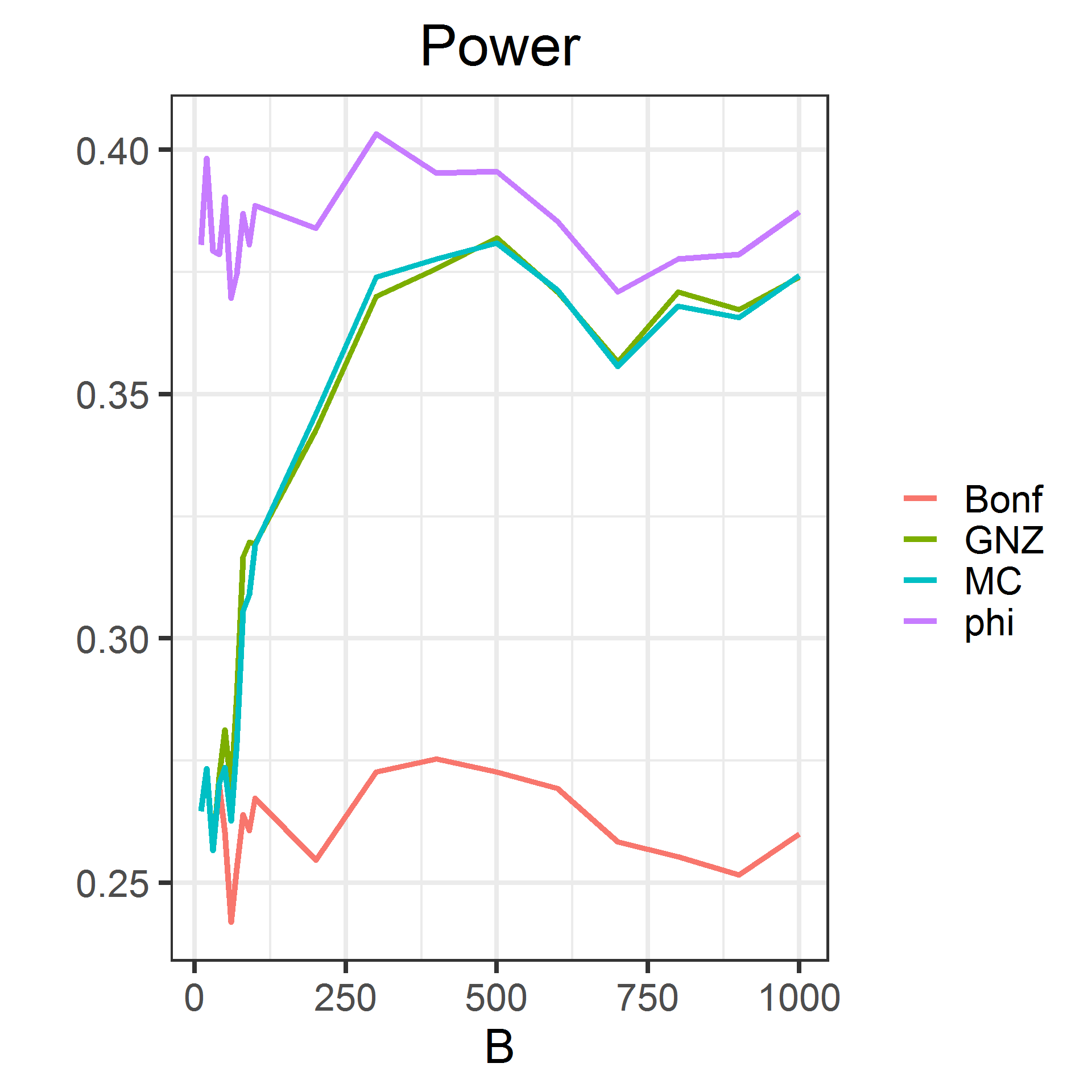}
\par\end{centering}
\caption{The mean local significance level and mean empirical power over $L=1{,}000$
tests with sample size $B\in\left\{ 10,20,\ldots,100,200,\ldots,1000\right\} $.
All other parameters are set to default.\label{fig: Simulation 2-1}}
\end{figure}
\begin{figure}[H]
\begin{centering}
\includegraphics[width=0.4\columnwidth]{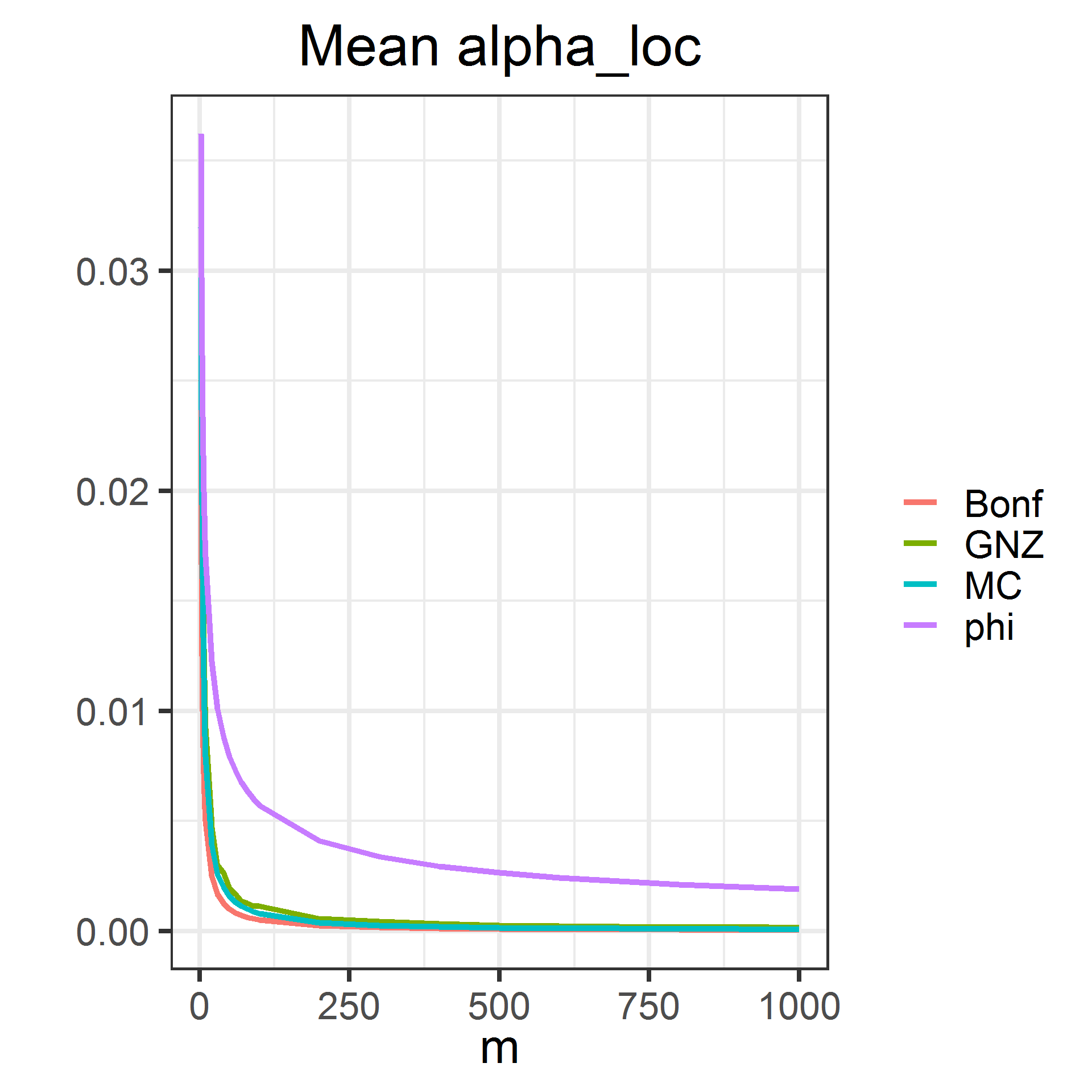}$\qquad$\includegraphics[width=0.4\columnwidth]{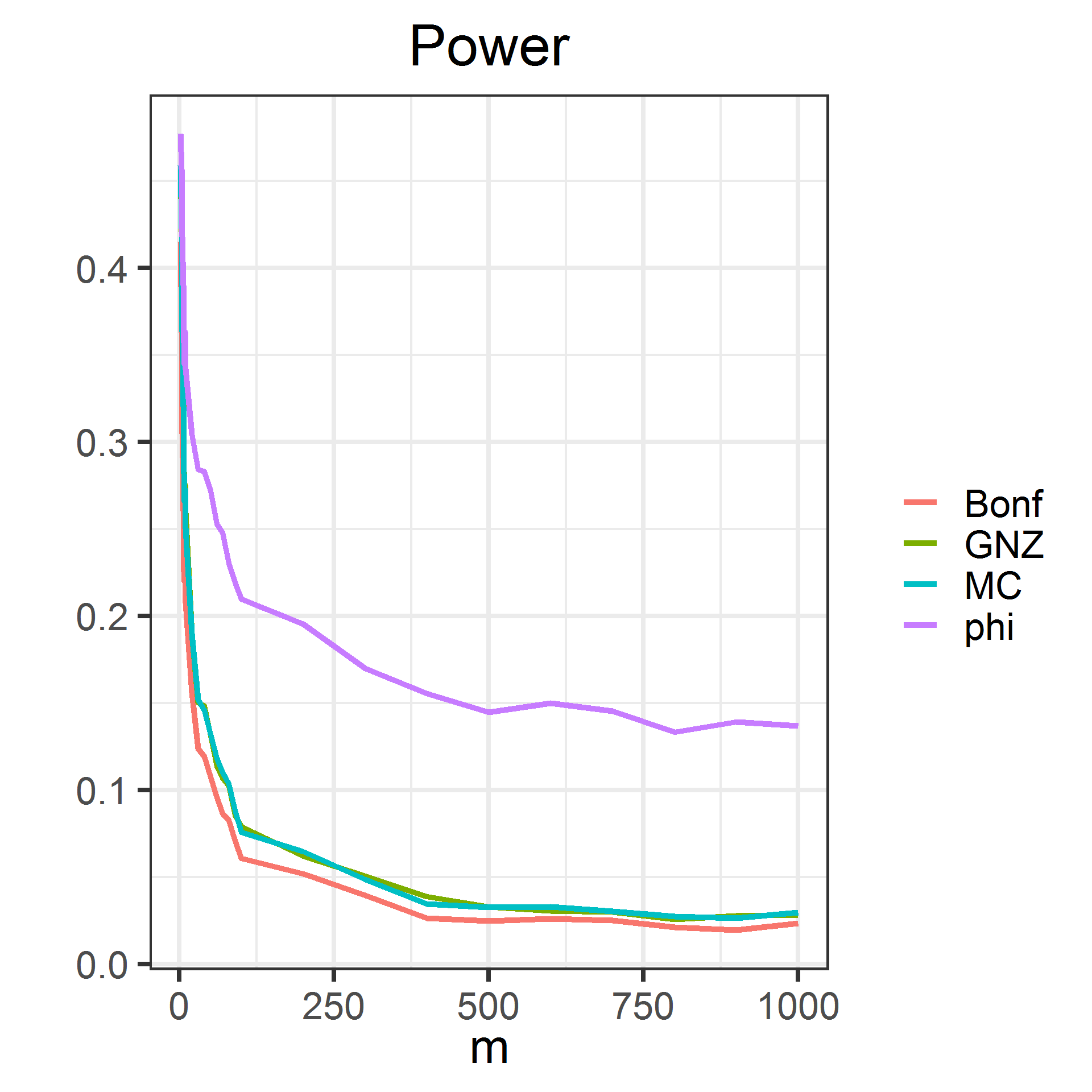}
\par\end{centering}
\caption{The mean local significance level and mean empirical power over $L=1{,}000$
tests with dimension $m\in\left\{ 2,3,\ldots,10,20,\ldots,100,200,\dots,1000\right\} $.
All other parameters are set to default.\label{fig: Simulation 2-2}}
\end{figure}
\begin{figure}[H]
\begin{centering}
\includegraphics[width=0.4\columnwidth]{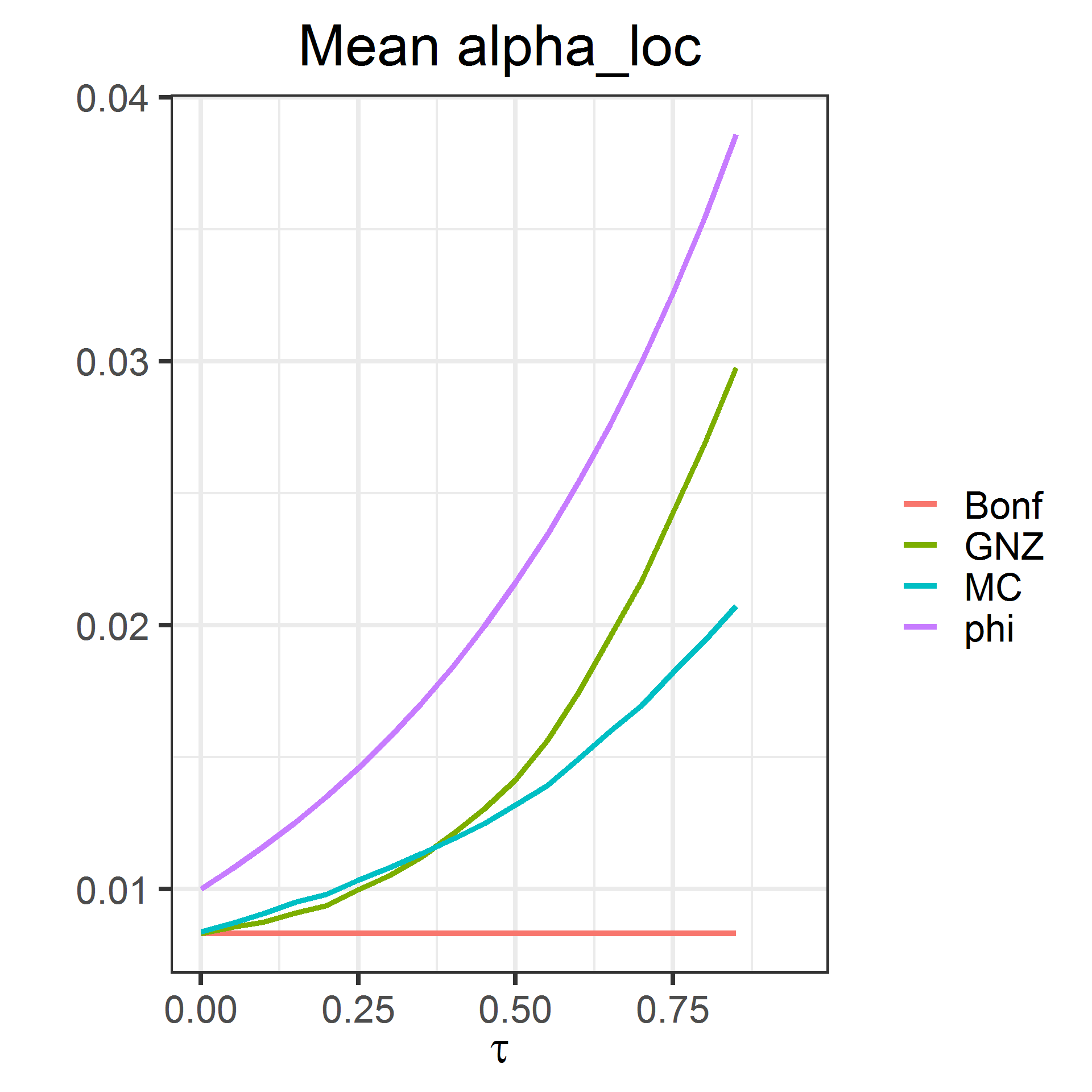}$\qquad$\includegraphics[width=0.4\columnwidth]{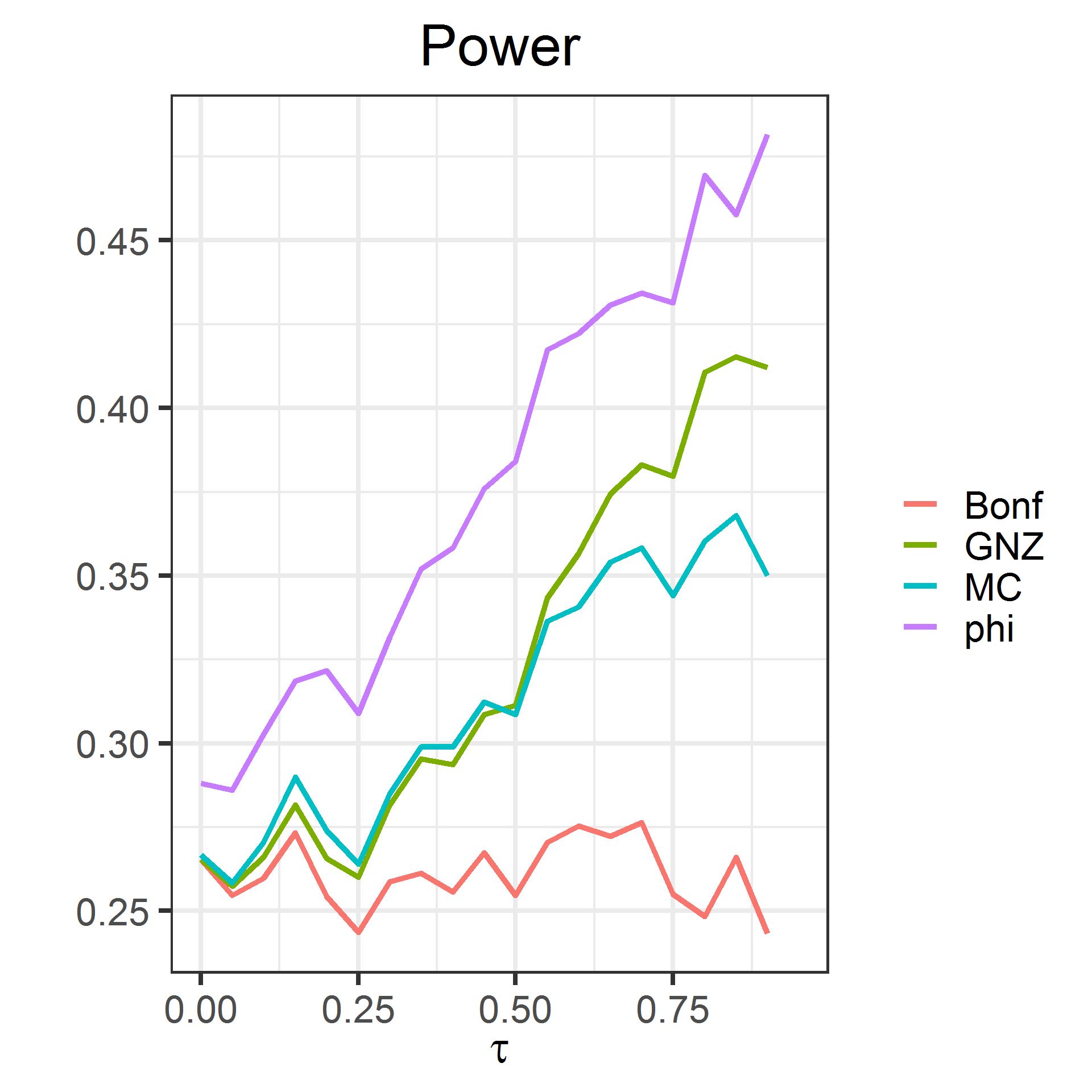}
\par\end{centering}
\caption{The mean local significance level and mean empirical power over $L=1{,}000$
tests with Kendall's $\tau\in\left\{ 0,0.05,\ldots,0.90\right\} $.
All other parameters are set to default.\label{fig: Simulation 2-3}}
\end{figure}

\subsubsection*{Simulation 3 - calibration in a misspecified setting}

In this simulation the data $\boldsymbol{X}_{1},\ldots,\boldsymbol{X}_{n}$
is sampled from a multivariate normal distribution with equi-correlation
coefficient $\rho$ and known variances equal to one (i.e., the covariance
matrix $\Sigma=\rho\boldsymbol{1}_{m\times m}+\left(1-\rho\right)\boldsymbol{I}_{m\times m}$,
where $\boldsymbol{I}_{m\times m}$ is the identity matrix). We test
the mean vector $\boldsymbol{\vartheta}=\boldsymbol{\mu}$ by the
hypotheses $H_{j}:\left\{ \mu_{j}=0\right\} $ versus $K_{j}:\left\{ \mu_{j}\neq0\right\} $.
The test statistic $T_{j}$ is given by $\left|\overline{\boldsymbol{X}_{j}}\right|=\left|n^{-1}\sum_{i=1}^{n}X_{i,j}\right|$
and the two-sided $p$-value $p_{j}$ is equal to $1-F_{T_{j}\mid H_{j}}\left(T_{j}\right)=2\left(1-F_{\overline{\boldsymbol{X}_{j}}\mid H_{j}}\left(T_{j}\right)\right)$.
We have $\overline{\boldsymbol{X}_{j}}\mid H_{j}\sim\mathcal{N}\left(0,n^{-1}\right)$
and $\overline{\boldsymbol{X}_{j}}\mid K_{j}\sim\mathcal{N}\left(\mu,n^{-1}\right)$
, where $\mu$ denotes the global effect size. Then $\boldsymbol{\mu}^{*}\equiv0$
is the LFC since the Gaussian copula does not depend on $\boldsymbol{\mu}$
and the LFC is marginally attained at zero. A sample of the test
statistics $\boldsymbol{T}_{1}^{*},\ldots,\boldsymbol{T}_{B}^{*}$
is calculated by bootstrapping the data. Under the LFC, $T_{b,j}^{*}$
is equal to $\left|\overline{\boldsymbol{X}_{j}^{*}}-\hat{\mu}_{j}\right|=\left|\overline{\boldsymbol{X}_{j}^{*}}-\overline{\boldsymbol{X}_{j}}\right|$.
We use the bootstrap sample of the test statistics to estimate the
generators. The estimation procedure is then the same as in the first
two simulations.

The default setting is $\tau=0.5$, $B=100$, $n=100,$ $m=6$ and
$\mu=0.2$. In \prettyref{fig: Simulation 3-1} and \prettyref{fig: Simulation 3-2},
the mean empirical power over $L$ tests resulting from the estimated
local significance level are compared. The empirical FWER is omitted
and (almost) always below $5\%$. The only exceptions are due to numerical
issues. With increasing $m$, the \verb-R- function \verb-uniroot-
tends to fail in this setup. This function is used for the estimator
$\hat{\alpha}_{\text{loc},B}^{\text{GNZ}}$. We skipped these cases
such that the results regarding GNZ in the right plot of \prettyref{fig: Simulation 3-1}
are potentially averaged over less than $L$ tests and can contain
numerical errors. In most settings the empirical power is close to
the Bonferroni method. A notable improvement is achieved under large
equi-correlation $\rho$. Overall, this simulation unfortunately suggests
that the Bonferroni correction is preferable when the assumption that
the test statistics copula is Archimedean is violated.
\begin{figure}[H]
\begin{centering}
\includegraphics[width=0.4\columnwidth]{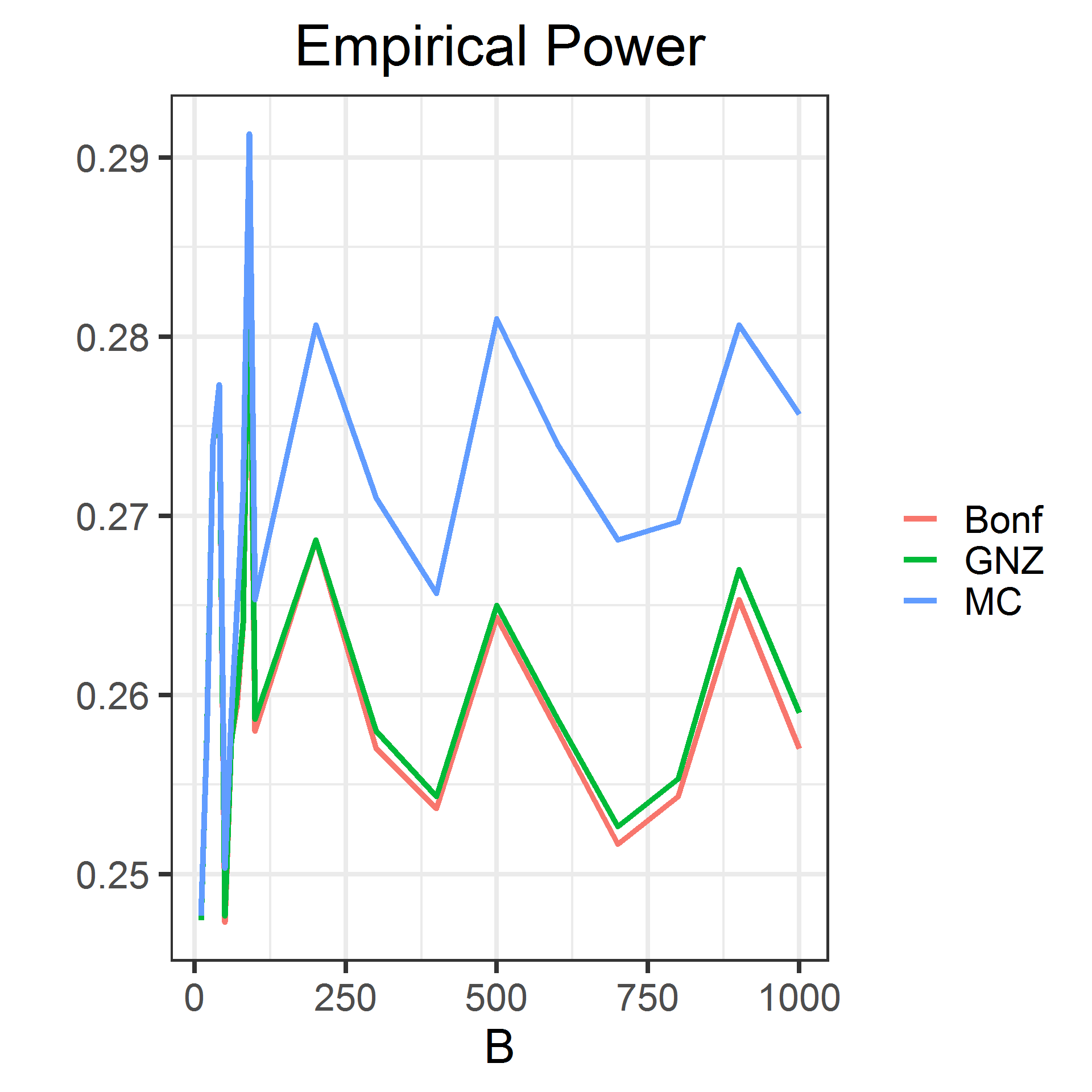}$\qquad$\includegraphics[width=0.4\columnwidth]{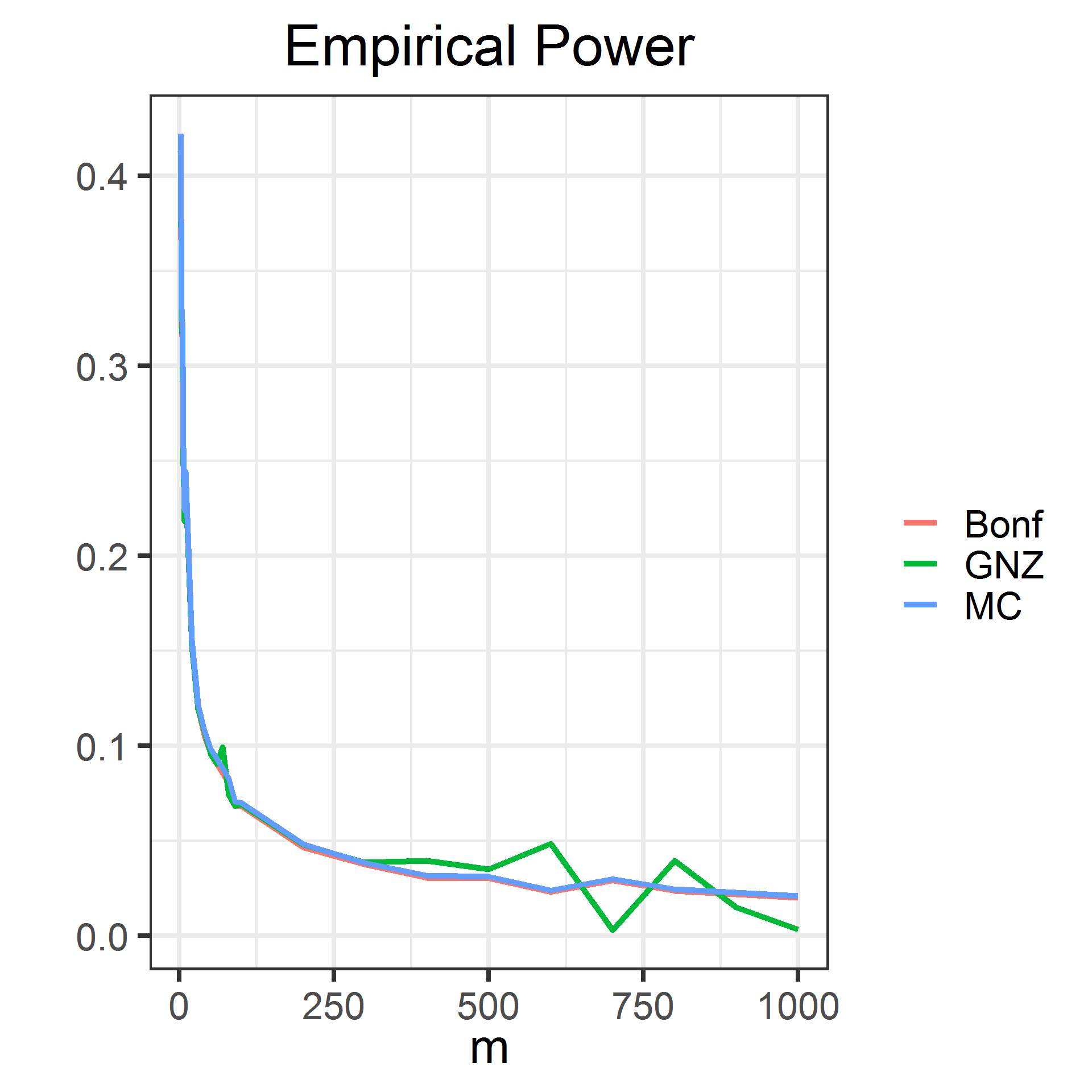}
\par\end{centering}
\caption{The mean empirical power over $L=1{,}000$ tests with sample size
$n\in\left\{ 10,20,\ldots,100,200,\ldots,1000\right\} $ (left) and
dimension $m\in\left\{ 2,3,\ldots,10,20,\ldots,100,200,\dots,1000\right\} $
(right). All other parameters are set to default.\label{fig: Simulation 3-1}}
\end{figure}
\begin{figure}[H]
\begin{centering}
\includegraphics[width=0.4\columnwidth]{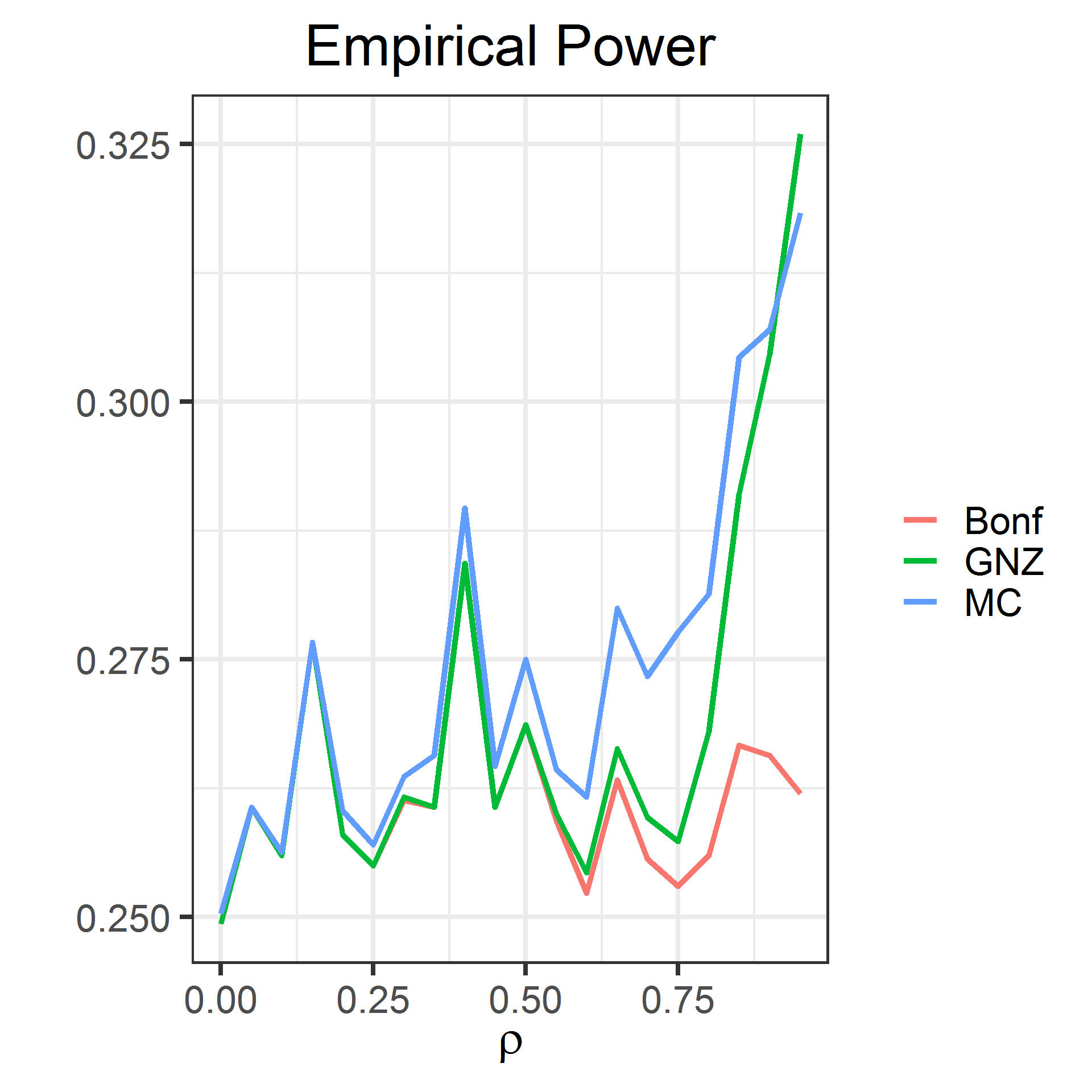}$\qquad$\includegraphics[width=0.4\columnwidth]{figures/2019-02-26-simulation-B-power}
\par\end{centering}
\caption{The mean empirical power with equi-correlation coefficient $\rho\in\left\{ 0,0.05,\ldots,0.95\right\} $
(left) and bootstrap sample size $B\in\left\{ 10,20,\ldots,100,200,\ldots,1000\right\} $
(right). All other parameters are set to default.\label{fig: Simulation 3-2}}
\end{figure}

\section{Conclusions\label{sec: Conclusions}}

We considered two non-parametric Archimedean generator estimator in
the context of multiple testing. The consistency of the modified estimator
is established and can be observed in \prettyref{fig: Simulation 1-1}
and \prettyref{fig: Simulation 2-1} as well. 

The simulations suggest that both estimators perform comparable but
the modified estimator is easier to implement and numerically more
stable. Some improvements of this estimator which are visible when
comparing both estimators directly do not carry over to a better performance
in multiple testing. For example, the modified estimator performs
considerably better under independence but the performance of both
resulting multiple tests is almost identical to the Bonferroni test
in this situation.

The modified estimator is still affected by the curse of dimensionality.
However, we did not perform separate sample size calculations given
the dimension and power. By Simulation 2, it is likely that the sample
size needs to be much larger than the dimension for a performance
comparable with the true generator. For example, in \prettyref{fig: Simulation 2-1}
the dimension is $6$ and the sample size needs to be larger than
$250$.

We observed that the Archimedean assumption is essential unless strong
dependencies are present. However, this might be only because the
dependency is not utilized by the Bonferroni test. We did not perform
comparisons with other copula estimation methods. Hence, we can not
generally conclude that multiple tests should be performed using non-parametric
Archimedean generator estimators under strong dependencies.

\subsubsection*{Supplementary Materials}

The source code, plots and numerical results of the simulations are
available online.


\end{document}